\pdfoutput=1
\documentclass[12pt,a4paper]{article}
\usepackage{amsmath,amssymb,amsfonts,mathtools,amsthm}
\usepackage{hyperref}
\usepackage{color}

\usepackage{graphicx}
\usepackage{todonotes}
\usepackage{comment}
\usepackage{todonotes}
\usepackage{authblk}
\usepackage{cite}

\newcommand{\ket}[1]{\ensuremath{|#1\rangle}}
\newcommand{\bra}[1]{\ensuremath{\langle#1|}}
\newcommand{\ketbra}[2]{\ensuremath{\ket{#1}\bra{#2}}}
\newcommand{\proj}[1]{\ensuremath{\ket{#1}\bra{#1}}}
\newcommand{\braket}[2]{\ensuremath{\langle{#1}|{#2}\rangle}}

\newcommand{\Id}{\mathrm{I}}
\newcommand{\ii}{\mathrm{i}}

\newcommand{\ee}{\mathrm{e}}

\newcommand{\Hl}{\mathcal H}

\newcommand{\GG}{\mathcal G}

\newcommand{\E}{\mathbb E}
\newcommand{\snode}{w}

\newcommand{\erfc}{\operatorname{erfc}}

\newcommand{\ie}{\emph{i.e.\/}}

\newtheorem{theorem}{Theorem}
\newtheorem{lemma}{Lemma}

\newtheorem{proposition}[theorem]{Proposition}

\begin{document}
	\title{Vertices cannot be hidden from quantum spatial search for almost all
random graphs}
	\author[1,2]{Adam Glos}
	
	\author[1,3]{Aleksandra Krawiec}
	
	\author[1,3]{Ryszard Kukulski\footnote{rkukulski@iitis.pl}}
	
	
	\author[1,4]{Zbigniew Pucha\l{}a}
	
	\affil[1]{Institute of Theoretical and Applied Informatics, Polish Academy
		of Sciences,  Ba{\l}tycka 5, 44-100 Gliwice, Poland}
	\affil[2]{Institute of Informatics, Silesian University of Technology,
			ul. Akademicka 16, 44-100 Gliwice, Poland}
	\affil[3]{Institute of Mathematics, University of Silesia, ul. Bankowa 14, 
		40-007 Katowice, Poland}
	\affil[4]{Faculty of Physics, Astronomy and Applied Computer Science,
		Jagiellonian University, ul. {\L}ojasiewicza 11, 30-348 Krak\'ow, Poland}
	\date{}


	\maketitle
	
	\begin{abstract}
	In this paper we show that all nodes can be found optimally for almost all
	random Erd\H{o}s-R\'enyi $\GG(n,p)$ graphs using continuous-time quantum
spatial search
	procedure. This works for both adjacency and Laplacian matrices, though under
	different conditions. The first one requires $p=\omega(\log^8(n)/n)$, while the
	seconds requires $p\geq(1+\varepsilon)\log (n)/n$,  where  $\varepsilon>0$. The
proof was made by
	analyzing the convergence of eigenvectors corresponding to outlying eigenvalues
	in the $\|\cdot\|_\infty $ norm. At the same time for
	$p<(1-\varepsilon)\log(n)/n$, the property does not hold for
	any matrix, due to the connectivity issues. Hence, our derivation concerning
	Laplacian matrix is tight.
\end{abstract}

\paragraph{Introduction}	Quantum walk is a topic of great interest in quantum information theory
\cite{childs2004spatial,childs2003exponential,ambainis2003quantum}. Numerous
possible applications were already discovered, including quantum
spatial search \cite{childs2004spatial,chakraborty2016spatial}, Google
algorithm \cite{paparo2012google,paparo2013quantum,sanchez2012quantum} or
quantum transport \cite{mulken2007quantum,mulken2011continuous}. Throughout
this article we consider quantum spatial search procedure, which is an example
of an algorithm yielding a result up to quadratically faster than its classical
counterpart.
Since the very first paper describing it was published
\cite{childs2004spatial}, plenty of new results have appeared in the
literature. This includes the noise resistance \cite{roland_noise_2005},
efficiency
analysis~\cite{childs2004spatial,chakraborty2016spatial,PhysRevLett.119.220503,tulsi_success_2016,philipp_continuous-time_2016},
imperfect implementation \cite{wong_spatial_2016}, difference in implementation
\cite{wong_laplacian_2016}, etc.

Unfortunately, most of the results concern very specific graph classes like
complete graphs \cite{childs2004spatial,roland_noise_2005} or their simplex
\cite{wong_spatial_2016}, and binary trees \cite{philipp_continuous-time_2016}.
Due to some kind of `symmetry', it was not necessary to make analysis for all
vertices separately (as for example in complete graphs or hypercubes), or at
least it could be easily fixed (for example by the level in binary trees). The
first big step towards the generalization into a large collection of graphs is
the work of Chakraborty et al. \cite{chakraborty2016spatial}, where
Erd\H{o}s-R\'enyi random graph model $\GG(n,p)$ was analyzed 
(with $n,p$ standing for the number of vertices and probability of an edge 
being present respectively). The authors have proven that
for almost all graphs almost all vertices can be found optimally. Since there
are already known examples of graphs for which some vertices are searched in
$\Theta(n^{\frac{1}{2}+a})$ time for 
$a>0$~\cite{childs2004spatial,philipp_continuous-time_2016} (throughout this 
paper
$O,o,\Omega,\omega,\Theta,\sim$ denote asymptotic relations, see
\cite{knuth1976big}), the result cannot be strengthened into `all graphs'.

The proof of the main result of Chakraborty et al. in
\cite{chakraborty2016spatial} is based on a lemma describing limit behavior of a
principal eigenvector $\ket{\lambda_1}$ of the adjacency
matrix. The authors show that for $p>\log^{\frac{3}{2}}(n)/n$, if $\ket
s=\frac{1}{\sqrt{n}}\sum_{v\in V}\ket v= \alpha \ket{\lambda_1}+\beta
\ket{\lambda_1}^\bot$, then almost surely $\alpha = 1-o(1)$. Since the time
needed for quantum spatial search is
$\Theta(\frac{1}{|\braket{\snode}{\lambda_1}|})$,  where $
\snode$ denotes the marked vertex, we have that almost all vertices can be
found in optimal time. 
However, in this case it is not trivial which vertex is chosen, since the
Erd\H{o}s-R\'enyi graph is not necessarily symmetric. This kind of
convergence allows the existence of vertices, which can be found in linear
time. As an example consider a vector
\begin{equation}
\ket{\lambda_1'}=\frac{1}{n \sqrt{k}} \sum_{i=1}^k \ket{i} + \frac{\sqrt{n^2 -
1}}{n \sqrt{n-k}}\sum_{i=k+1}^{n}\ket i , \label{eq:chakraborty-counterexample-state}
\end{equation} for $k=o(n)$.
We have, $\braket{s}{\lambda_1'} = \alpha =  1 - o(1)$ and thus a-priori,
vector $\ket{\lambda_1'}$
from~\cite{chakraborty2016spatial}, can be the leading eigenvector of an 
adjacency matrix.
In such a case the argument used by \cite{chakraborty2016spatial} is not tight 
enough to exclude a possibility that all 
vertices $w \in
\{1,\dots,k\}$ will be found in $\Omega(n)$ time, which is actually a random
guess complexity. Note that it is even possible that for almost all graphs such
vertices exist. Furthermore, many of the applications mentioned in
\cite{chakraborty2016spatial} require $\braket{i}{v_1}\approx\braket{j}{v_1}$
for arbitrary $i,j$. Otherwise, creating Bell states or quantum transport will
be at least very difficult.

What is more, due to the laws of quantum mechanics, the measurement time needs
to be known since the beginning. This includes not only differences in the
complexity, but a constant as well. For example, if for two different nodes
$v,v'$ we have $\braket{v}{\lambda_1}=\frac{1}{\sqrt n}$  and
$\braket{v'}{\lambda_1}=\frac{2}{\sqrt n}$, then different measurement times
should be chosen for each.

Both effects mentioned above can be described as hiding nodes in the graphs.
Finally, we propose a following research problem: \emph{can we actually `hide' a
	vertex in a random Erd\H{o}s-R\'enyi graph}? We have managed to show, that 
	in
the case of adjacency matrix, $p=\omega(\log^3(n)/(n\log^2(\log(n)))$ is a 
sufficient
requirement for all-vertices optimal search. Under further constraint 
$p=\omega(\log^8(n)/n)$
we have common time measurement. Moreover, we went a step 
further than the authors of \cite{chakraborty2016spatial} and studied also 
Laplacian matrix, which led us to tighter 
results. In the case of Laplacian matrix,
$p>(1+\varepsilon)\log(n)/n$, for constant $\varepsilon>0$, 
is sufficient for common time measurement, however
in the $p=\Theta(\log(n)/n)$  case, it may not be true that almost surely the
probability 1 of a successful measurement is achieved. If
$p<(1-\varepsilon)\log(n)/n$, then a random graph contains almost surely
isolated nodes \cite{erdos1960evolution}, hence it is possible to hide a vertex.

\paragraph{Element-wise optimality for adjacency matrix.} Let $G=(V,E)$
be a simple undirected graph with node set $V=\{1,\dots,n\}$ and edge set
$E\subset V\times V$. Moreover, let $\Hl_G$ be a quantum system spanned by an
orthonormal basis $\{\ket v : v \in V\}$. Quantum spatial search is based on the
Schr\"odinger differential equation
\begin{equation}
\ket{\dot \psi_t} = -\ii H \ket{\psi_t} = -\ii \left ( -M_G- \proj {\snode}
\right ) \ket{\psi_t},\label{eq:quantum-spatial-search}
\end{equation}
where  $M_G$ is a matrix corresponding to the graph structure, typically
rescaled adjacency matrix $A$ or Laplacian $L=D-A$, where $D$ 
is the degree matrix.
In \cite{chakraborty2016spatial}, authors have proven that for a random
Erd\H{o}s-R\'enyi graph in case of an adjacency matrix almost all vertices from
almost all graphs can be found optimally. We say some property holds almost 
surely for all graphs, where the probability of choosing random graph having 
such is $1-o(1)$.
The result was based on the following simplified lemma.
\begin{lemma}[\cite{chakraborty2016spatial}] \label{lem:the-lemma}
	Let $H$ be a Hamiltonian with eigenvalues $\lambda_1\geq\dots\geq\lambda_n$ 
	satisfying $\lambda_1=1$ and $|\lambda_i|\leq c <1$ for all $i>1$ with 
	corresponding eigenvectors $\ket {\lambda_1}=\ket 
	s,\ket{\lambda_2},\dots,\ket{\lambda_n}$ and let $w$ denote a marked 
	vertex. For an appropriate choice of 
	$r\in[-\frac{c}{1+c},\frac{c}{1-c}]$, the starting state $\ket{s}$ evolves 
	by the Schr$\ddot{o}$dinger's equation with the Hamiltonian $(1+r)H 
	+\ketbra{w}{w}$ for time $t = \Theta(\sqrt n)$ into the state $\ket f $ 
	satisfying $|\braket{w}{f}|^2\geq \frac{1-c}{1+c}+o(1)$.
\end{lemma}
According to the proof of the lemma, the bound can be derived by choosing $r$ 
satisfying
\begin{equation}
\sum_{i=2}^n \frac{|\braket{w}{\lambda_i}|^2}{(1+r)\lambda_i -r} = 
\sum_{i=2}^{n}|\braket{w}{\lambda_i}|^2 \label{eq:equality-r}.
\end{equation}
The assumptions from the Lemma guarantee the existence of $r\in 
[\frac{-c}{1+c},\frac{c}{1-c}]$ satisfying the above  equality. Note that the 
result is constructive for $c=o(1)$, as in this case $r=o(1)$ as well as 
$t=\frac{\pi \sqrt{n}}{2}$. Otherwise, a proper determination of $r$ and $t$ is 
needed.
	
	According to the Lemma, two properties of $M_G$ are useful in proving search
optimality. Firstly, the matrix should have a single outlying eigenvalue.
Secondly, if $\ket{\lambda_1}$ is the eigenvector corresponding to the outlying
eigenvalue, one should have $ \vert \braket{\snode}{\lambda_1} \vert
=\Theta(\frac{1}{\sqrt n})$.

Note that in the limit $n\to\infty$, norms cease to be equivalent, thus
different concepts of closeness of vectors can be chosen. In
\cite{chakraborty2016spatial}, authors choose $1-|\braket{\psi}{\phi}|$
for arbitrary vectors $\ket{\psi}, \ket{\phi}$, which allows
to infer
that $o(n)$ of nodes can be found in time $\omega(\sqrt n)$, 
see the example given in Eq. \eqref{eq:chakraborty-counterexample-state}. 
In order to make statements concerning all vertices we should 
study
the limit behavior of the principal vector in $L^{\infty}$ norm $\| \cdot
\|_\infty$, 
which bounds the maximal deviation of coordinates.
More precisely, we are interested whether $\|\ket{
\lambda_{1}}-\ket {s}\|_\infty
= \frac{o(1)}{\sqrt n}$, as this would imply that for an arbitrary marked node $
\snode$ we have $ \braket{\snode}{\lambda_{1}}=(1+o(1))\frac{1}{\sqrt n}$.
The above will give us the bound 
$\Theta(\frac{1}{|\braket{\snode}{\lambda_1}|}) 
= \Theta(\sqrt{n})$ for 
the time needed for quantum spatial algorithm to locate vertex $w$.

	Indeed, 
	a convergence of infinity norm
	was shown by Mitra~\cite{mitra_entrywise_2009}
	providing $p\geq \log^6(n)/n$. We have managed to weaken the assumptions and
	thereby strengthen the result.
	\begin{proposition}
		Suppose $A$ is an adjacency matrix of a random Erd\H{o}s-R\'enyi graph
$\GG(n,p)$
		with $p=\omega( \log^3(n)/(n\log^2\log n))$. Let $\ket{\lambda_{1}}$ denote
the eigenvector
		corresponding to the largest eigenvalue of $A$ and let $\ket s=\frac{1}{\sqrt
			n}\sum_{v}\ket v$. Then 
		\begin{equation}
		\|\ket{\lambda_{1}}-\ket s \|_\infty =o\left (\frac{1}{\sqrt n} \right )
		\end{equation}
		with probability $1-o(1)$.
	\end{proposition}
	The proof, which follows the concept proposed by Mitra
	\cite{mitra_entrywise_2009}, can be found in
	Section~\ref{app:element-wise-bound} in Supplementary Materials. This implies
	that all vertices can be found optimally in $\Theta(\sqrt{n})$ time for almost
	all graphs.
	
	To show the common time measurement suppose that the largest eigenvalue of
$\frac{1}{np}A$ satisfies
	\begin{equation}
	\left| \lambda_{1}-1\right|\leq \delta,
	\end{equation}
	where $\delta\to0$. Then the probability of measuring the  searched vertex $w$
in time $t$ can be approximated by~\cite{chakraborty2016spatial}
	\begin{equation}
	\begin{split}
	P_\omega(t) &= \left| \bra w \exp(-\ii Ht) \ket s\right|^2  \\
	&\approx\frac{1}{1+n\delta^2/4}\sin^2 \left(\sqrt{\delta^2/4 + 1/n} t\right)
\label{eq:probability-measurement}.
	\end{split}
	\end{equation} 
	Since $\ket{\lambda_1}$ tends to $\ket s$ in the $\|\cdot\|_\infty $ norm, the
	approximation works for all nodes. Hence, when $\delta=O(\frac{1}{\sqrt{n}})$
	(with small constant in the $\Theta(\frac{1}{\sqrt{n}})$ case), then all of the
	vertices can be found in time $O(\sqrt{n})$. Nevertheless, $\delta$ depends on
a
	chosen graph, and thus the measurement time may differ. In order to ensure that
	the time and probability of measurement are the same for all marked nodes and
almost all graph chosen, one
	should provide $\delta  = o(\frac{1}{\sqrt{n}})$ almost surely. 
	
	If $p=\omega(\log^8(n)/n)$, then the largest eigenvalue $\lambda_1$ follows
	$\mathcal N\left( 1,\frac{1}{n}\sqrt{2(1-p)/p} \right)$ distribution
	\cite{erdhos_spectral_2013}, see Section~\ref{app:largest-eigenvalue-adjacency}
	in Supplementary Materials for a step-by-step derivation, where $\mathcal
	N(\mu,\sigma )$ is the normal distribution with mean $\mu$ and standard
	deviation~$\sigma$. Therefore, one can show that  asymptotically almost surely
	\begin{equation}
	\left | \lambda_{1}\left(\frac{1}{np} A\right) -1 \right |\leq \delta,
	\end{equation}
	where $\delta=o(\frac{1}{n\sqrt{p}})$. Note that since $np=\omega(\log ^8(n))$,
	we have actually $\delta=o(\frac{1}{\sqrt n\log^4(n)})$ in the worst case
	scenario. This, in turn, allows us to use the simplified version of
	Eq.~\eqref{eq:probability-measurement}
	\begin{equation}
	P_{\omega}(t) \approx \sin^2\left(\frac{t}{\sqrt n}\right)
	\end{equation}
	for large $n$. Thus we have that in time $t=\frac{\pi}{2}\sqrt n$, the
	probability of measurement is optimal, independently on a  chosen marked node.
	Finally, we can conclude our results concerning adjacency matrix with the
	following theorem.
	\begin{theorem}
		Suppose we chose a graph according to Erd\H{o}s-R\'enyi $\GG(n,p)$ model with
		$p=\omega(\log^8(n)/n)$. Then by choosing  $M_G=\frac{1}{np}A$, where $A$ is
an
		adjacency matrix in Eq.~\eqref{eq:quantum-spatial-search}, almost surely all
		vertices can be found with probability $1-o(1)$  with common measurement time
approximately $t=\pi\sqrt{n}/2$.
	\end{theorem}

\paragraph{Element-wise optimality for Laplacian matrix.} Similar
property
	holds for a Laplacian matrix $L$. This is a positive semi-definite matrix,
where
	the dimensionality of null-space corresponds to the number of connected
	components.  Based on the results from \cite{chung_spectra_2011}, one can show
	that for $p=\omega(\log (n) /n)$ all of the others eigenvalues of
$\frac{L}{np}$
	converge to 1, see Section \ref{app:Laplacian-spectrum} in Supplementary
	Materials. At the same time, the  eigenvector corresponding to the null-space
is
	\emph{exactly} the equal superposition $\ket s = \ket {\mu_n} = \frac{1}{\sqrt
		n}\sum_{v\in V} \ket v$. Thus, since for $p>(1+\varepsilon)\log(n)/n$ a graph
is
	almost surely connected, the Laplacian matrix takes the form
	\begin{equation}
	\frac{L}{np} = 0 \cdot  \proj{s} + \sum_{i=1}^{n-1}\mu_i \proj{\mu_i},
	\end{equation}
	where $\mu_i\to 1$ almost surely for $1\leq i\leq n-1$. Here
	$\mu_1,\dots,\mu_{n}$ denote eigenvalues of Laplacian matrix with corresponding
	eigenvectors $\ket{\mu_1},\dots,\ket{\mu_n}$. 
	Note that since the identity
	matrix corresponds to global phase change only, which is an unmeasurable
	parameter, we can equivalently choose
	\begin{equation}
	\Id-\frac{L}{np} = \proj s- \sum_{i=1}^{n-1}(\mu_i-1) \proj{\mu_i}.
	\end{equation}
	Note that the matrix above satisfies the requirements of Lemma~1 from
	\cite{chakraborty2016spatial}, and therefore all of the vertices can be found
	optimally with probability $1-o(1)$. Common time measurement is a direct
	application of Lemma~1 from \cite{chakraborty2016spatial}, since more in-depth
	proof analysis shows, that under the theorem assumptions, $t=\pi\sqrt n /2$
	should be chosen for maximizing the success probability.
	
	The situation changes in the case of $p=O(\log(n)/n)$. Note that for both
	adjacency and Laplacian matrices the evolution does not change the probability
	of measuring isolated vertices. If $p<(1-\varepsilon)\log(n)/n$, then graphs
	almost surely contain such vertices and hence you actually \emph{can} hide a
	vertex in such a graph.
	
	The $p\sim p_0\log (n)/n$ for a constant $p_0>1$ is a smooth transition case between
hiding
	and non-hiding cases mentioned before. In this case based on the Exercise III.4
	from \cite{bollobas2001random}, one can show that $\mu_1 \sim
	(1-p_0)(W_{0}(\frac{1-p_0}{\ee p_0}))^{-1}\log(n)$ and  $\mu_{n-1} \sim
	(1-p_0)(W_{-1}(\frac{1-p_0}{\ee p_0}))^{-1}\log(n)$, where $W_{0},W_{-1}$ are
	Lambert W functions, see Section \ref{app:largest-eigenvalue-laplacian} in
	Supplementary Materials. Here we use the notation $f(n)\sim g(n) \iff
	f(n)-g(n)=o(g(n))$ . In this case, the $M_G=\Id-\frac{1}{np}L$ does not imply
	that both $\mu_1$ and $\mu_{n-1}$ converge to 1.
	
	Nevertheless, we can still make simple changes in a matrix in order to obtain
	optimality of the procedure. Let $a=(1-p_0)(W_{0}(\frac{1-p_0}{\ee p_0}))^{-1}
$
	and $b=(1-p_0)(W_{-1}(\frac{1-p_0}{\ee p_0}))^{-1}$ denote constants
	corresponding to $\mu_{1}$ and $\mu_{n-1}$ limit behavior. Then
	\begin{equation}
	\Id-\frac{2}{(a+b)\log n}L \label{eq:rescaled-laplacian-treshold}
	\end{equation}
	again satisfies Lemma~1 from \cite{chakraborty2016spatial} with
	$c=\frac{a-b}{2}$. According to the Lemma~1, the probability of success after
	time $t=\frac{\pi}{2 \sqrt n}$ is bounded from below by
	\begin{equation}\label{eq:p_bound}
	p_{\mathrm{bound}}=\frac{1-c}{1+c} = {W_{0}\left(\frac{1-p_0}{\ee
			p_0}\right)}/{W_{-1}\left(\frac{1-p_0}{\ee p_0}\right)}.
	\end{equation}
	The bound converges to 0 when $p_0\to 1^+$ and to 1 when $p_0\to\infty$, and
	monotonically changes in $(1,\infty)$, see Fig.~\ref{fig:treshold-probability}.
	Note that this corresponds to the other results. For $p_0<1$ the
	probability of measuring all vertices is equal to 0 due to the connectivity
	issues mentioned before. For $p_0\to \infty$, the situation becomes similar to
	$p=\omega(\log(n)/n)$, where non-hiding property was already shown. Note
however, that the actual success probability seems to be much higher than the
bound, see Fig.~\ref{fig:success-probability-true}. Eventually,
	we conclude all of the results by the following theorems. 
	
	\begin{figure}\centering
		\includegraphics[scale=0.95]{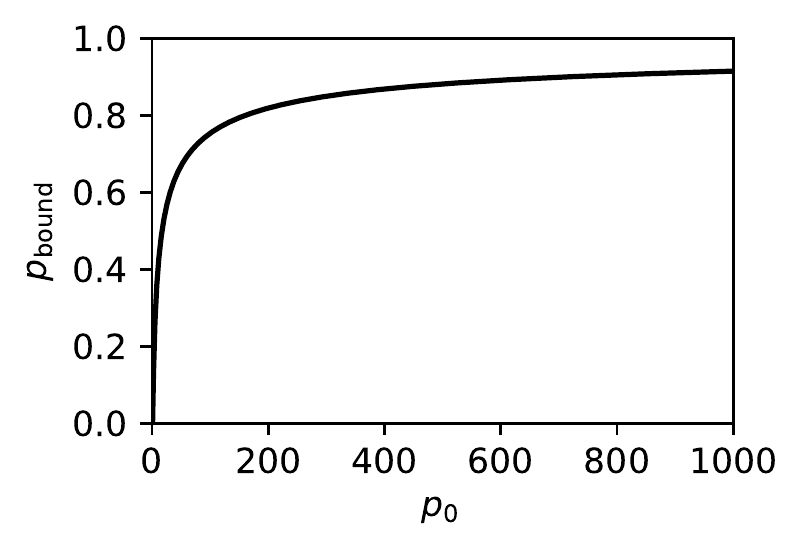} \caption{The lower
			bound of success probability of quantum spatial search for
			$p=p_0\frac{\log(n)}{n}$ for almost all graph from Erd\H{o}s-R\'enyi graph
			model. The exact formula is $p_\mathrm{bound}=W_{0}\left(\frac{1-p_0}{\ee
				p_0}\right)/W_{-1}\left(\frac{1-p_0}{\ee p_0}\right)$. Note that
			$p_{\mathrm{bound}}\to 0$ as $p_0\to 1^+$, where connectivity threshold is
			achieved. Furthermore $p_{\mathrm{bound}}\to 1$ as $p_0\to \infty$}
		\label{fig:treshold-probability}
	\end{figure} 
	
	
	\begin{theorem}
		Suppose we chose a graph according to Erd\H{o}s-R\'enyi $\GG(n,p)$ model. For
		$p=\omega(\log(n)/n)$, by choosing $M_G=\frac{1}{np}L$ in
		Eq.~\eqref{eq:quantum-spatial-search}, almost surely all vertices can be found
		with probability $1-o(1)$ in asymptotic $\pi\sqrt n/2$ time. For $p\sim
		p_0\log(n)/n $, by choosing $M_G = (1+r)\gamma L$ for some proper $r$, where
$\gamma$ is defined as in
		Eq.~\eqref{eq:rescaled-laplacian-treshold}, all vertices can be found in
		$\Theta(\sqrt{n})$ time with probability bounded from below by the
		constant in Eq. \eqref{eq:p_bound}. 
	\end{theorem}
	We leave determining proper $r$ and $t$ values as open question. 
	
	\begin{figure}\centering
		\begin{minipage}{\textwidth}\centering
			\includegraphics{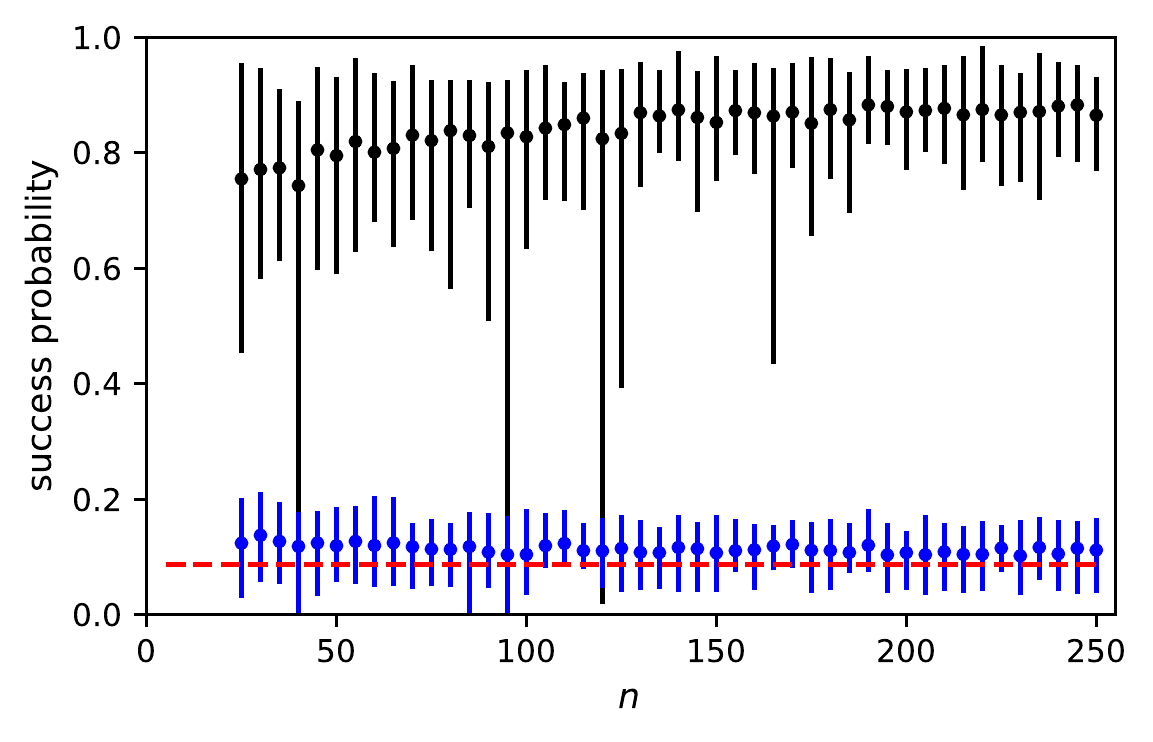}
		\end{minipage}\\
		\begin{minipage}{\textwidth}\centering
		\includegraphics{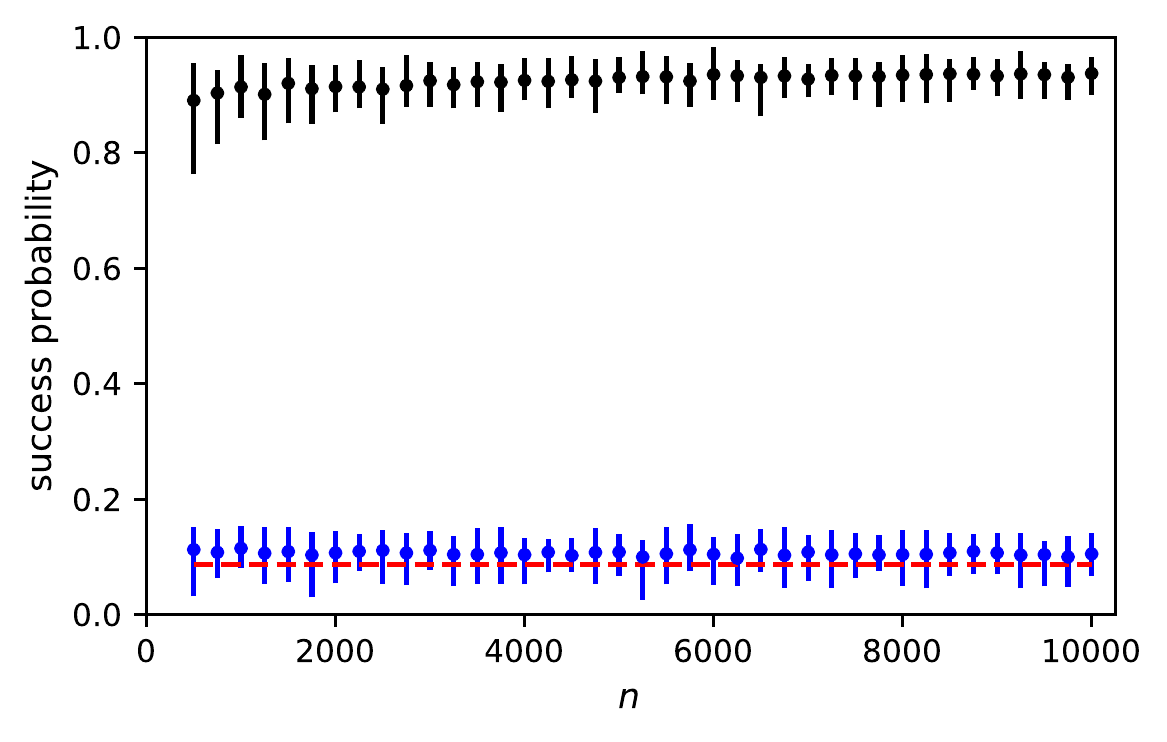}
		\end{minipage}
	\caption{The figure presents probability bounds for quantum spatial search
together with success probability derived from simulation. The red dashed line
denotes the limit bound for succes probability. The blue error bars denote
$\frac{1-c}{1+c}$ for $c=\max\{|\lambda_2|,|\lambda_n|\}$ for matrix from
Eq.~\ref{eq:rescaled-laplacian-treshold} for randomly chosen graph.  Black error
bars denote the actual success probability. Deviations correspond to the maximal
and minimal obtained values. Graphs were chosen according to the
$\GG(n,2\frac{\log (n)}{n})$ model, $r$ were derived according to the
Eq.~\ref{eq:equality-r} and we chose time $t=\frac{\pi \sqrt n}{2}$. 30 graphs
were chosen for each size. One can see that the bound for randomly chosen graph
oscillates around the limit value, nonetheless the true success probability  is
much higher than the bound.} \label{fig:success-probability-true}
	\end{figure}
	
	\begin{theorem}
		Suppose we chose a graph according to Erd\H{o}s-R\'enyi $\GG(n,p)$ model with
		$p\leq (1-\varepsilon)\log(n)/n$, where $\varepsilon >0$. Then for both
		adjacency and Laplacian matrices there exist vertices which cannot be found in
$o(n)$ time.
	\end{theorem}

	%
	%
	
	\paragraph*{Conclusion and discussion.}
	
	In this work we prove that all vertices can be found optimally with common
	measurement time $(\pi\sqrt n)/2$  for almost all Erd\H{o}s-R\'enyi graphs
	for both adjacency and Laplacian matrices under conditions
	$p=\omega(\log^8(n)/n)$ and $p\geq(1+\varepsilon)\log(n)/n$ respectively. The
	proof is based on element-wise ergodicity of the eigenvector corresponding to
the 
	outlying eigenvalue of adjacency or Laplacian matrix. While under the mentioned
	constraint adjacency matrix almost surely  achieves success probability
	$1-o(1)$, the same probability for Laplacian matrix in the $p\sim p_0\log
(n)/n$
	case for some $p_0>1$ can only be bounded from below by some positive constant.
	At the same time for $p<(1-\varepsilon)\log (n)/n$, the property does not hold
	anymore, since almost surely there exist isolated vertices which need
	$\Omega(n)$ time to be found.
	
	While our derivation concerning the Laplacian matrix is nearlt complete, since
only upper-bound for success probability is missing in the $p=\Theta(\log(n)/n)$
case, in our opinion
	it is possible to weaken the condition on $p$ for the adjacency matrix. The
	first key step would be showing that the largest eigenvalue
	$\lambda(\frac{1}{np} A)$ follows $\mathcal N(1,\frac{1}{n}\sqrt{2(1-p)/p)}$
	distribution for $p\geq(1+\varepsilon)\log(n)/n$. Then, since element-wise
	convergence of principal vector requires $p=\omega(\log^3(n)/(n\log^2\log n))$,
the result
	would be strengthened to the last mentioned constraint. The second step would
be the
	generalization of the mentioned element-wise convergence theorem.
	
	
	Further interesting generalization of the result would be the analysis of more
	general random graph models as well. While this proposition has already been
	stated \cite{chakraborty2016spatial}, our results show that in order to prove
	security of the quantum spatial search, it would be desirable to analyze the
	limit behavior of the principal vector in the sense of $\|\cdot\|_\infty$ norm.

	{\noindent \bf {Acknowledgements} }
	Aleksandra Krawiec, Ryszard Kukulski and Zbigniew Pucha\l{}a acknowledge the
support from the National Science Centre, Poland 
	under project number 2016/22/E/ST6/00062. Adam Glos
	was supported by the National Science Centre under project number
DEC-2011/03/D/ST6/00413.

\appendix
\newpage

\section{Element-wise bound on principal
	eigenvector}\label{app:element-wise-bound} Let $G_{n,p}$ be a random
Erd\H{o}s-R\'{e}nyi graph, $\deg(v)$ be a degree of the vertex $v \in V$ and $A$
be its adjacency matrix with eigenvalues $\lambda_1 \geq \lambda_2 \geq \ldots
\geq \lambda_n$. Let also $\ket{\lambda_i}$ be an eigenvector corresponding to
the eigenvalue $\lambda_i$ and $\ket{s} = \frac{1}{\sqrt{n}}\ket{\mathbf{1}}
=\frac{1}{\sqrt{n}} \sum_{i=1}^n \ket{v}  $.

\begin{proposition}
For the probability $p = \omega \left( \ln^3(n)/(n\log^2\log n) \right)$and some constant $c>0$ we have
\begin{equation}
\Vert \ket{\lambda_1} -\ket{s} \Vert_{\infty} \leq c \frac{1}{\sqrt{n}}\frac{\ln^{3/2}(n)}{\sqrt{np} \ln(np)}
\end{equation}
with probability $1-o(1)$.
\end{proposition}

\begin{proof}
Using  \cite{chung_spectra_2011}, we have
\begin{gather}
\Vert A-E(A) \Vert \leq \sqrt{8np\ln(n)}, \label{eq:zal1} \\
|\lambda_{1} -np| \leq \sqrt{8np\ln(\sqrt{2}n)}, \label{eq:zal2} \\
\underset{ i\geq 2}{\max} |\lambda_i| \leq \sqrt{8np\ln(\sqrt{2}n)}, \label{eq:zal3}
\end{gather}
with probability $1-o(1)$. The first inequality was shown in the proof of
Theorem 1 while the second and third inequalities come from Theorem 3 in 
\cite{chung_spectra_2011}. Note $\deg(v)$ follows a binomial distribution. Using
Lindenberg's CLT and the fact that the convergence is uniform one can show that
\begin{equation}
\begin{split}
P\left( |\deg(v)-np| \leq 2\sqrt{\ln(n)np(1-p)} \right) &
\approx P\left(|\mathcal X| \leq 2\sqrt{\ln(n)}\right) \\
&\geq 1- \frac{1}{\sqrt{2\pi\ln(n)} n^2},
\end{split}
\end{equation}
where $\mathcal X$ is a random variable with standard normal distribution. Let
$A=\lambda_1\ketbra{\lambda_1}{\lambda_1}+\sum_{i\geq 2}
\lambda_i\ketbra{\lambda_i}{\lambda_i}$ and $\ket{s} = \alpha
\ket{\lambda_1}+\beta \ket{\lambda_1^\perp}$. Assume that
$\ket{\lambda_1},\ket{\lambda_1^\perp},\ket{\lambda_i}$ are normed vectors and
$\ket{\lambda_1^\perp}=\sum_{i \geq 2} \gamma_i \ket{\lambda_i}$. By the
Perron-Frobenius Theorem we can choose a vector $\ket{\lambda_1}$ such that
$\braket{v}{\lambda_1}\geq 0$ and hence obtain $\braket{s}{\lambda_1}=\alpha>0$.
Thus 
\begin{equation}
\begin{split}
&\left(A-E(A)\right)\ket{\lambda_1}\\
&=\left(\lambda_1\ketbra{\lambda_1}{\lambda_1}+\sum_{i\geq
 2} 
\lambda_i\ketbra{\lambda_i}{\lambda_i}-np\ketbra{s}{s}\right)\ket{\lambda_1}\\ 
&= (\lambda_1 -np\alpha^2)\ket{\lambda_1}-np\alpha \beta 
\ket{\lambda_1^{\perp}}.
\end{split}
\end{equation}
With probability $1-o(1)$, using Eq.~(\ref{eq:zal1}) we have
\begin{equation}
\begin{split}
(\lambda_1 -np\alpha ^2)^2+(np)^2\alpha ^2\beta ^2&=\Vert \left(A-E(A)\right)\ket{\lambda_1} \Vert^2 \\&\leq 8np \ln(n)
\end{split}
\end{equation}
and thus since $\beta^2=1-\alpha^2$, then
\begin{equation}
\alpha^2np(np-2 \lambda_1)+\lambda_1^2 \leq 8np \ln(n).
\end{equation}
Eventually, we receive 
\begin{equation}
\begin{split}
1 \geq \alpha \geq \alpha^2 &\geq \frac{\lambda_1^2-8np \ln(n)}{2 \lambda_1 np-(np)^2}  \\  &\geq 1-\frac{4}{2+\sqrt{\frac{np}{8\ln(\sqrt{2}n)}}} \\
&\geq 1-\frac{16}{\sqrt{\frac{np}{\ln(n)}}},
\label{eq:szac.a}
\end{split}
\end{equation}
where the fourth inequality comes from Eq.~(\ref{eq:zal2}). We know that
$|\deg(v)-np| \leq 2 \sqrt{n \ln(n) p(1-p)}$ with probability greater than
$1-\frac{1}{n^2}$. Thus, with probability $1-\frac{1}{n}$ the above is true for
all $v \in V$ simultaneously. Now, since $\deg(v)= \bra{v} A \ket{\mathbf{1}}$, 
we have
\begin{equation}
\begin{split}
\frac{np - 2 \sqrt{n \ln(n) p(1-p)}}{\lambda_1} &
\leq \frac{1}{\lambda_1}\bra{v}\texttt{}A\ket{\mathbf{1}} \\
&\leq  \frac{np + 2 \sqrt{n \ln(n) p(1-p)}}{\lambda_1}
\end{split}
\end{equation}
The lower bound can be estimated as
\begin{equation}
\begin{split}
\frac{np - 2 \sqrt{n \ln(n) p(1-p)}}{\lambda_1} 
&\overset{(\ref{eq:zal2})}{\geq}  \frac{1-2\sqrt{\ln(n)\frac{1-p}{np}}}{1+\sqrt{8\frac{\ln(\sqrt{2}n)}{np}}} \\
&\geq 
\frac{1-2\sqrt{\frac{\ln(n)}{np}}}{1+4\sqrt{\frac{\ln(n)}{np}}} \eqqcolon d
\end{split}
\end{equation}
and similarly the upper bound
\begin{equation}
\frac{np +  2 \sqrt{n \ln(n) p(1-p)}}{\lambda_1} \leq \frac{1+2\sqrt{\frac{\ln(n)}{np}}}{1-4\sqrt{\frac{\ln(n)}{np}}} \eqqcolon u.
\end{equation}
Consequently 
\begin{equation}
\frac{d}{\sqrt n}\leq \frac{1}{\lambda_1}\bra{v} A\ket{s} \leq \frac{u}{\sqrt n} \label{eq:13}
\end{equation}
for all $v \in V$. Let $l=c\frac{\ln(n)}{\ln(\sqrt{\frac{np}{\ln(n)}}/4)}$, where $c=c(n,p) \in [1, 2)$ is chosen to satisfy $l=\left\lceil  \frac{\ln(n)}{\ln(\sqrt{\frac{np}{\ln(n)}}/4)} \right\rceil$. Hence
\begin{equation}
\frac{d^l}{\sqrt n}\leq \bra{v} \left(\frac{A}{\lambda}  \right)^l\ket{s}\leq \frac{u^l}{\sqrt n}
\end{equation}
for all $v \in V$.
On the other hand
\begin{equation}
\begin{split}
\left(\frac{1}{\lambda_1}A \right)^l \left(\alpha \ket{\lambda_1}+\beta \ket{\lambda_1^\perp}\right)&=\left(\ketbra{\lambda_1}{\lambda_1}+\sum_{i\geq 2} \left(\frac{\lambda_i}{\lambda_1}\right)^l \ketbra{\lambda_i}{\lambda_i}\right)\left(\alpha \ket{\lambda_1}+\beta \ket{\lambda_1^\perp}\right)\\ 
&=\alpha\ket{\lambda_1}+\beta \sum_{i\geq 2} \left(\frac{\lambda_i}{\lambda_1}\right)^l \gamma_i \ket{\lambda_i}.
\end{split}
\label{eq:14}
\end{equation}
Using Eq.~(\ref{eq:zal1},\ref{eq:zal2}) we are able to estimate  $\frac{\lambda_i}{\lambda_1}$ by 
\begin{equation}
\begin{split}
\frac{\lambda_i}{\lambda_1} &\leq
\frac{\sqrt{8np\ln(\sqrt{2}n)}}{np-\sqrt{8np\ln(\sqrt{2}n)}}\\
&=\frac{1}{\sqrt{\frac{np}{8\ln(\sqrt{2}n)}}-1} \leq \frac{4}{\sqrt{\frac{np}{\ln(n)}}}.
\end{split}
\end{equation}
Thus
\begin{equation}
\begin{split}
\left\Vert \beta \sum_{i\geq 2} \left(\frac{\lambda_i}{\lambda_1}\right)^l \gamma_i \ket{\lambda_i} \right\Vert_\infty 
&\leq |\beta| \left\Vert  \sum_{i\geq 2} \left(\frac{\lambda_i}{\lambda_1}\right)^l \gamma_i \ket{\lambda_i} \right\Vert_2\\
&\leq |\beta| \sqrt{ \sum_{i\geq 2} \gamma_i^2\left(\frac{4}{\sqrt{\frac{np}{\ln(n)}}}\right)^{2l}} \\&= \frac{|\beta|}{\left(\frac{\sqrt{\frac{np}{\ln(n)}}}{4}\right)^{l}}
= \frac{|\beta |}{n^c}\\
&\leq \frac{4}{\left(\frac{np}{\ln(n)}\right)^{1/4} n},
\end{split}
\label{eq:16}
\end{equation}
where the last inequality comes from Eq.~(\ref{eq:szac.a}) and \mbox{$\|\cdot\|_2$} denotes the Euclidean norm.
By Eq.~(\ref{eq:13},\ref{eq:14}) we get
\begin{equation}
\frac{d^l}{\sqrt n}\leq \alpha\braket{v}{\lambda_1}+ \bra{v} \left( \beta \sum_{i\geq 2} \left(\frac{\lambda_i}{\lambda_1}\right)^l \gamma_i \ket{\lambda_i} \right)  \leq\frac{u^l}{\sqrt n},
\end{equation}
for all $v \in V$ and using Eq.~(\ref{eq:szac.a},\ref{eq:16}) we eventually obtain
\begin{equation}
\frac{\frac{d^l}{\sqrt{n}}-\frac{4}{\left(\frac{np}{\ln(n)}\right)^{1/4} n}}{1} \leq \braket{v}{\lambda_1}  \leq
 \frac{\frac{u^l}{\sqrt{n}}+\frac{4}{\left(\frac{np}{\ln(n)}\right)^{1/4} n}}{1-\frac{16}{\sqrt{\frac{np}{\ln(n)}}}}
\end{equation}
for all $v \in V$. In order to finish the proof it is necessary to show that
\begin{equation}
(1-d^l)+\frac{4}{\left(\frac{np}{\ln(n)}\right)^{1/4} \sqrt{n}} = O\left(\frac{\ln^{3/2}(n)}{\sqrt{np} \ln(np)}\right)
\label{eq:21}
\end{equation}
and
\begin{equation}
(u^l-1)+\frac{4}{\left(\frac{np}{\ln(n)}\right)^{1/4} \sqrt{n}}= O\left(\frac{\ln^{3/2}(n)}{\sqrt{np} \ln(np)}\right).
\end{equation}
We need to estimate how quickly  $d^l$ converges to $1$. Using the fact that $d
\rightarrow 1$, it is enough to observe that
\begin{equation}
(1-d)l = O\left(\frac{\ln^{3/2}(n)}{\sqrt{np} \ln\left(  \sqrt{\frac{np}{\ln(n)}}/4\right)} \right),
\end{equation}
and thus 
\begin{equation}
1-d^l \approx 1 - e^{(d-1)l} = O\left(\frac{\ln^{3/2}(n)}{\sqrt{np} \ln(np)} \right).
\end{equation}
The second term of LHS of Eq.~(\ref{eq:21}) converges to $0$ more rapidly than
the bound, so it completes the proof for the lower bound. The same thing for the
upper bound can be shown analogously.
\end{proof}

\section{Distribution of the largest eigenvalue of adjacency matrix}
\label{app:largest-eigenvalue-adjacency} Theorem~6.2 from
\cite{erdhos_spectral_2013} considers the distribution of the largest eigenvalue
of rescaled adjacency matrix $\tilde A = A/ \sqrt{(1-p)pn}$. They show that as
long as $p>\frac{1}{n}$, then
\begin{equation}
\E \lambda_{1} (\tilde  A) = \sqrt\frac{np}{1-p}+\sqrt\frac{1-p}{np} +o(1).\label{eq:max-eigenvalue-normed-exp}
\end{equation}
Furthermore, under another condition $p=\omega(\log^8(n)/n)$ we have
\begin{equation}
\sqrt{\frac{n}{2}}\left( \lambda_{1}(\tilde A)-\E \lambda_{1}(\tilde A)\right) \to \mathcal N(0,1)\label{eq:max-eigenvalue-normed-dist}
\end{equation}
in a distribution. This allows us to derive the distribution of the largest
eigenvalue of the  $\frac{1}{np} A$ matrix
\begin{equation}
\begin{split}
\lambda_{1}\left (\frac{1}{np} A \right ) &= \sqrt\frac{1-p}{np} \lambda_{1}(\tilde A) \\\
&= \frac{\sqrt{2(1-p)}}{n\sqrt p} \left(\sqrt\frac{n}{2}\left( \lambda_{1}(\tilde A)-\E \lambda_{1}(\tilde A)\right)   + \sqrt\frac{n}{2}\E \lambda_{1}(\tilde A) \right) \\
&= \frac{\sqrt{2(1-p)}}{n\sqrt p}  \mathcal X  + 1 + \frac{1-p}{np} +\sqrt\frac{1-p}{np}o(1)
\end{split}
\end{equation}
where $\mathcal X \to \mathcal{N}(0,1)$. Hence we have that
$\lambda_{1}(\frac{1}{np} A) \sim \mathcal
N(1,\frac{1}{n}\sqrt\frac{2(1-p)}{p})$. Note, that under the condition
$p=\omega(\log^8(n)/n)$, the standard deviation tends to~0. This means that the
largest eigenvalue actually tends to the Dirac distribution $\delta_{x=1}$.

This gives as a bound for $\lambda_{1}(\frac{1}{np} A)$. Note that 
\begin{equation}
\mathbb P(|\lambda_{1}( A/(np))-1 |\leq \delta) = 1 -  \erfc \left(\frac{n \sqrt p\delta}{2\sqrt{1-p}}\right).
\end{equation}
The probability tends to 1 as long as the argument tends to $\infty$. In order
to achieve this, we need to assume $n\sqrt{p}\delta \to \infty$ as $n\to\infty$.
This can be done by choosing $\delta = o( \frac{1}{n\sqrt{p}})$. Eventually, we
have asymptotically almost surely
\begin{equation}
\left | \lambda_{1}\left(\frac{1}{np}A\right) -1 \right |= o\left( \frac{1}{n\sqrt{p}} \right).
\end{equation}
 Note that for $p=o(1)$ the bound is better than the one used in \cite{chakraborty2016spatial}.

\section{Laplacian matrix spectrum} \label{app:Laplacian-spectrum} Algebraic
connectivity satisfies $\mu_{n-1}=np+O(\sqrt{np\log n})$ for
$p=\omega(\log(n)/n)$. Similarly we conclude from results of Bryc et al.
\cite{bryc_spectral_2006}, that $\mu_1 \sim np$.

\begin{theorem}
	Let $L_n$ be a Laplacian matrix of random Erd\H{o}s-R\'enyi graph $\GG(n,p)$,
where $p=\omega(\frac{\log n}{n})$. Then $\mu_1=\mu(L) \sim np$.
\end{theorem}
\begin{proof}
By Theorem~1.5 from \cite{bryc_spectral_2006}, if $\tilde L$ is a symmetric
matrix whose off-diagonal elements have two-points distribution with mean 0 and
variance $p(1-p)$ and $\tilde{L}_{ii} =  \sum_{j \neq i} \tilde{L}_{ij}$. Then
\begin{equation}
\lim_{n\to\infty} \frac{\mu(\tilde L)}{\sqrt{2np(1-p)\log n}}=1. \label{eq:max-eigenvalue-laplacian}
\end{equation}
Note that in the following version $p$ may depend on $n$. Hence, we can extend
the Corollary 1.6 from the same paper.

Let $L_n= \tilde L_n + Y_n$, where $Y_n$ is a deterministic matrix with $-p$ on
off-diagonal and $(n-1)p$ on diagonal. Note that $Y_n$ is an expectation of a
random Erd\H{o}s-R\'enyi Laplacian matrix. $Y_n$ has a single 0 eigenvalue and
all of the others take the form $np$. By this we have $\mu(Y_n)= \|Y_n\| = np$.
Then we have
\begin{equation}
\left|\frac{\|L_n\|}{np}-\frac{\|Y_n\|}{np}\right| \leq \frac{\|L_n-Y_n\|}{np} = \frac{\|\tilde L_n\|}{np}\to 0
\end{equation}
where the limit comes from the Eq.~\eqref{eq:max-eigenvalue-laplacian}, assuming $p=\omega(\log n/n)$. Finally $\frac{\mu(L_n)}{np} \to 1$.
\end{proof}

\section{The largest eigenvalue of Laplacian matrix near the connectivity treshold}\label{app:largest-eigenvalue-laplacian}

Suppose $G$ is a random graph chosen according to $\GG (n,p_0\frac{\log(n)}{n})$  distribution, for $p_0>1$ being a constant. It can be shown, that 
\begin{equation}
\delta \sim (1-p_0) \left(W_{-1}\left  (\frac{1-p_0}{\ee p_0}\right)\right)^{-1}\log(n)
\end{equation}
and
\begin{equation}
\Delta \sim (1-p_0) \left(W_{0}\left  (\frac{1-p_0}{\ee p_0}\right)\right)^{-1}\log(n),
\end{equation}
see \cite{bollobas2001random}, Exercise III.4. Here $\delta$ and $\Delta$ denote
respectively minimal and maximal degree of the graph. In
\cite{kolokolnikov_algebraic_2014} authors have shown that providing
\begin{equation}
|\delta - cnp| =O(\sqrt{np})
\end{equation}
we have
\begin{equation}
|\mu_{n-1} - cnp| =O (\sqrt{np}),
\end{equation}
where $\mu_{n-1}$ is the second smallest eigenvalue of the Laplacian matrix. In
fact, similar behavior can be stated for the largest eigenvalue, \ie{} if
\begin{equation}
|\Delta- cnp| = o(np)
\end{equation}
we have
\begin{equation}
|\mu_{1} - cnp| = o(np).
\end{equation}
While we plan to prove the statement above, it is possible that the RHS can be
reduced to $O(\sqrt{np})$ by following the proof in
\cite{kolokolnikov_algebraic_2014}. Nonetheless, we are satisfied with the
mentioned result. The proof is very similar to the proof of Lemma 3.4 in
\cite{kolokolnikov_algebraic_2014}. Furthermore, note, that the theorem holds
for $p_0>0$.
\begin{theorem}
	Suppose there exists a $p_0>0$ so that $np\geq p_0\log(n)$ and $\Delta\sim cnp$	almost surely. Then almost surely $\mu_{1} \sim cnp$.
\end{theorem} 
\begin{proof}
	Note, that since the eigenvector corresponding to 0 eigenvalue is the equal superposition, we have
	\begin{equation}
	\begin{split}
	\mu_{1} &= \max_{\{\ket{\phi}\bot\ket s: \braket{\phi}{\phi}=1\}} \bra \phi L \ket \phi \\
	&= \max_{\{\ket{\phi}\bot\ket s: \braket{\phi}{\phi}=1\}} (\bra \phi D \ket \phi-\bra \phi A \ket \phi).
	\end{split}
	\end{equation}
	Note that 
	\begin{equation}
	\begin{split}
	\mu_{1 }&\leq \max_{\{\ket{\phi}\bot\ket s: \braket{\phi}{\phi}=1\}} \bra \phi D \ket \phi +\max_{\{\ket{\phi}\bot\ket s: \braket{\phi}{\phi}=1\}} |\bra \phi A \ket \phi)|\\
	&\leq \Delta + C\sqrt{np}
	\end{split}
	\end{equation}
	by Theorem 2.5 from \cite{feige2005spectral}. Similarly one can show $\mu_{1
}\geq \Delta$, which can be done by taking maximum over canonical vectors.
After combining those bounds and $\Delta \sim cnp$, we obtain the result.
\end{proof}

\end{document}